\documentclass[11pt]{article}
\usepackage[utf8]{inputenc}
\usepackage[T1]{fontenc}
\usepackage[top=1in, bottom=1in, left=1in, right=1in, letterpaper]{geometry}

\usepackage{dsfont}
\usepackage{amsmath}
\usepackage{hyperref}
\usepackage{amsthm}
\usepackage{tikz} 
\usepackage{algorithm}
\usepackage[noend]{algpseudocode}

\usepackage{cite}

\usepackage{thmtools} 
\usepackage{thm-restate}
\usepackage[inline]{enumitem}

\usepackage{color}

\definecolor{blueblack}{rgb}{0,0,.7}

\newcounter{sideremark}

\definecolor{Darkblue}{rgb}{0,0,0.4}
\definecolor{Brown}{cmyk}{0,0.61,1.,0.60}
\definecolor{Purple}{cmyk}{0.45,0.86,0,0}
\definecolor{brickred}{rgb}{0.8, 0.25, 0.33}

\newtheorem{theorem}{Theorem}

\newtheorem{lemma}{Lemma}
\newtheorem{claim}{Claim}

\newtheorem{observation}{Observation}

\newcommand{\eps}{\varepsilon}

\newcommand{\calT}{\mathcal{T}}
\newcommand{\calS}{\mathcal{S}}

\newcommand{\calP}{\mathcal{P}}

\newcommand{\OPT}{\mathsf{OPT}}
\newcommand{\opt}{\mathsf{opt}}

\newcommand{\APX}{\mathsf{APX}}

\newcommand{\selfloop}{\mathsf{sl}}

\newcommand{\floor}[1]{\left\lfloor #1 \right\rfloor}
\newcommand{\ceil}[1]{\left\lceil #1 \right\rceil}

\newcounter{casei}
\newcounter{caseii}[casei]
\newcounter{caseiii}[caseii]

\newenvironment{caseanalysis}{\setcounter{casei}{0}}{}

\newcommand{\case}[1]{%
\refstepcounter{casei}%
\smallskip\noindent\textbf{\boldmath(\arabic{casei}) #1}\unboldmath%
\phantomsection%
}
\newcommand{\subcase}[1]{%
\refstepcounter{caseii}%
\smallskip\noindent\textbf{\boldmath(\arabic{casei}.\arabic{caseii}) #1}\unboldmath%
\phantomsection%
}

\usepackage[colorinlistoftodos,textsize=tiny,textwidth=2cm,color=green!50!gray]{todonotes} %,disable

\title{A PTAS for Weighted Triangle-free 2-Matching}

\author{Miguel Bosch-Calvo\thanks{\texttt{\{miguel.boschcalvo,fabrizio\}@idsia.ch}, IDSIA, USI-SUPSI, Lugano, Switzerland.}
\and Fabrizio Grandoni\footnotemark[1] 
\and Yusuke Kobayashi\thanks{\texttt{\{yusuke,tnoguchi\}@kurims.kyoto-u.ac.jp}, RIMS, Kyoto University, Japan.}
\and Takashi Noguchi\footnotemark[2]
}

\date{}
\begin{document}

\maketitle

\begin{abstract}
\noindent In the \emph{Weighted Triangle-Free 2-Matching} problem (WTF2M), we are given an undirected edge-weighted graph. Our goal is to compute a maximum-weight subgraph that is a 2-matching (i.e., no node has degree more than $2$) and triangle-free (i.e., it does not contain any cycle with $3$ edges). One of the main motivations for this and related problems is their practical and theoretical connection with the Traveling Salesperson Problem and with some $2$-connectivity network design problems. 
WTF2M is not known to be NP-hard and at the same time no polynomial-time algorithm to solve it is known in the general case (polynomial-time algorithms are known only for some special cases). The best-known (folklore) approximation algorithm for this problem simply computes a maximum-weight 2-matching, and then drops the cheapest edge of each triangle: this gives a $2/3$ approximation. 
In this paper we present a PTAS for WTF2M, i.e., a polynomial-time $(1-\eps)$-approximation algorithm for any given constant $\eps>0$. Our result is based on a simple local-search algorithm and a non-trivial analysis.
\end{abstract}

\section{Introduction}\label{sec:intro}

Let $G=(V,E,w)$, $w\colon E\rightarrow \mathds{Q}_{\geq 0}$, be a simple undirected edge-weighted graph. A $2$-matching $M$ of $G$ is a subset of edges $M\subseteq E$ such that at most two edges of $M$ are incident on each node in $V$. Observe that a $2$-matching induces a collection of disjoint paths and cycles. The $2$-Matching problem requires to compute a $2$-matching of maximum weight $w(M):=\sum_{e\in M}w(e)$: this problem can be solved in polynomial time via a reduction to matching (see, e.g.~\cite[Chapter 30.6.]{S03}). 

In this paper we focus on the \emph{Weighted Triangle-Free 2-Matching} problem (WTF2M), a natural generalization of 2-Matching where we further require that the computed matching $M$ does not induce any \emph{triangle} (i.e., cycle with $3$ edges).
Hartvigsen proposed an (extremely complex) exact algorithm for the unweighted version TF2M in his Ph.D. thesis in 1984~\cite{HartD}, and a (similarly complex) journal version of that result was recently published~\cite{R-HartD}.
Recently, Paluch~\cite{PalK} reported another polynomial-time algorithm for the problem (still unpublished).
For the sake of simplification in verification and analysis, Bosch-Calvo, Grandoni, and Ameli~\cite{BGA} proposed a simple local search-based Polynomial-Time Approximation Scheme (PTAS) for this problem, i.e., an algorithm that, for any given input constant parameter $\eps>0$, computes a ($1-\eps$)-approximate solution for the problem in polynomial time.
Kobayashi and Noguchi~\cite{KN2025} further and substantially simplified the correctness proof of the mentioned PTAS. Exact algorithms and PTASs for TF2M are used as subroutines to design improved approximation algorithms for some network design problems (more details in Section \ref{sec:related}).

Much less is known for the weighted version of the problem. WTF2M is not-known to be NP-hard. Polynomial-time algorithms are known for some special cases of the problem. More specifically, Hartvigsen and Li~\cite{HL} showed that WTF2M can be solved in polynomial-time when the input graph is subcubic, i.e., the maximum degree is at most $3$ (see also \cite{BK} for generalizations and \cite{PW,K2010} for simplifications). Furthermore, Kobayashi~\cite{K2022} proposed a polynomial-time algorithm for the problem when all triangles in the input graph are edge-disjoint. The best-known (folklore) approximation algorithm for WTF2M simply computes a maximum-weight $2$-matching and then drops the cheapest edge of each triangle: this gives a $2/3$ approximation.

\subsection{Our Results and Techniques}

In this paper we present a PTAS for WTF2M, the first improvement over the mentioned trivial $2/3$ approximation.

\begin{theorem}\label{thr:ptas}
There is a PTAS for the Weighted Triangle-Free 2-Matching problem. 
\end{theorem}

In more detail, let $\OPT(G)$ denote some optimal solution to the considered WTF2M instance $G$, and $\opt(G)=w(\OPT(G))$ denote its weight. We simply use $\OPT$ and $\opt$ when $G$ is clear from the context. 

Our PTAS is based on a simple local-search algorithm, that works as follows.
Given a triangle-free $2$-matching $M$, \textit{an augmenting trail $P$ for $M$} is a trail such that $M\Delta P$ is a triangle-free $2$-matching of weight strictly larger than $M$. Here, $\Delta$ denotes the symmetric difference. Given a trail $P$, by $|P|$ we denote its number of edges. At each step our algorithm finds an augmenting trail $P$ with bounded $|P|$, if any, and uses it to increase the weight of the current solution $M$ (initially empty); see Algorithm \ref{alg:PTAS} for the details.

\begin{algorithm}
    \caption{Local Search for WTF2M.\\
    \textbf{Input:} A constant $\eps\in(0, 1]$ and an undirected edge-weighted graph $G=(V,E,w)$.\\
    \textbf{Output:} A triangle-free $2$-matching $\APX$ of $G$ satisfying $w(\APX)\geq (1-\eps)\opt(G)$.}
    \begin{algorithmic}[1]
    \State $\APX\gets \emptyset$
    \While{there exists an augmenting trail $P$ for $\APX$ with $|P|\leq 7/\eps$}
    \State $\APX\gets \APX\Delta P$
    \EndWhile
    \State \textbf{return} $\APX$
    \end{algorithmic}
    \label{alg:PTAS}
\end{algorithm}

The key to prove that the above algorithm works is the following theorem. 
\begin{theorem}\label{thm:main}
    Let $G = (V, E, w)$ be a simple weighted graph, $\APX$ be a triangle-free $2$-matching of $G$, and $\eps\in(0, 1]$ be a constant. If $(1-\eps)\opt(G) > w(\APX)$, then there exists an augmenting trail $P$ for $\APX$ with $|P|\leq 7/\eps$. 
\end{theorem}

Theorem~\ref{thm:main} implies that Algorithm~\ref{alg:PTAS} returns a $(1-\eps)$-approximate solution. However, it is not obvious that it runs in polynomial time. To overcome this difficulty, we reduce the problem to the case with polynomially bounded integer weights; see Section~\ref{sec:reduction} for details.

Our proof for Theorem~\ref{thm:main} follows a line similar to that of~\cite{KN2025}, which is based on graph reductions. The key ingredient is a generalization of triangle-freeness, which is called $\mathcal{T}$-freeness.
In more detail, we generalize the triangle-free constraint by introducing a list $\mathcal{T}$ of triangles that must not be included in the desired 2-matching.
By transforming the graph, we reduce the input graph to one with a smaller size of $\mathcal{T}$, and by transforming the trail obtained from the reduced graph back onto the original graph, we construct the required trail (see Section~\ref{sec:mainproof}). The desired result is then obtained by setting $\mathcal{T}$ to the set of all the triangles in $G$.

As we will see in the technical part, our approach also works for non-simple graphs that contain self-loops but not parallel edges. On such graphs, we obtain a PTAS for the problem of computing a $\mathcal{T}$-free $2$-matching for any input set $\mathcal{T}$ of triangles.

It remains an important open problem in the area to find an exact polynomial-time algorithm for WTF2M, if one exists.

\subsection{Related Work}
\label{sec:related}

We already mentioned the exact algorithms and PTASs for TF2M. 
An exact algorithm for TF2M was used as a subroutine to design improved approximation algorithms for the $\{1,2\}$-TSP problem, a well-studied special case of TSP where all edge weights are $1$ or $2$. A key observation here is that, given a maximum-cardinality $2$-matching $M$ on the edges of cost $1$, $2n-|M|$ is a lower bound on the minimum cost of a TSP tour. The same lower bound holds if one uses instead a maximum-cardinality triangle-free $2$-matching. In \cite{BS} a maximum-cardinality triangle-free $2$-matching is used as a starting point to build a feasible TSP tour. The triangle-freeness turns out to be a useful feature for technical reasons. A similar approach is also announced in \cite{AMP}.

An exact algorithm for TF2M is also used as a subroutine to design improved approximation algorithms for the $2$-Edge-Connected Spanning Subgraph problem (2ECSS). In this problem we are given an undirected graph, and the goal is to compute a subgraph with the minimum number of edges that contains all the nodes and that is $2$-edge-connected, i.e., it remains connected after removing one arbitrary edge. This problem is very well studied in terms of approximation algorithms \cite{KV,CSS,VV,HVV,SV,GGA}. The latest improved approximation algorithms \cite{HLL,BGGHJL25,KN} work as follows. They first build a maximum-cardinality triangle-free $2$-matching $M$. Then they convert $M$ into a minimum-cardinality triangle-free $2$-edge-cover\footnote{We recall that a 2-edge-cover is a subset of edges $C$ such that each node has at least $2$ edges of $C$ incident to it.} $C$. It can be shown that $|C|$ is a valid lower bound on the size of the optimal solution to the input 2ECSS instance. Finally, $C$ is used as a starting point to build a feasible 2ECSS solution. Also in this case, the lack of triangles in $C$ turns out to be a helpful feature in the construction for technical reasons. Indeed, it would be very helpful to avoid even longer cycles, say, with up to $4$ or $5$ edges. However no exact or sufficiently good approximation algorithms are known for that goal.  
For similar reasons, TF2M could also be used to simplify (and possibly improve) the approximation factor for the related $2$-Vertex-Connected Spanning Subgraph problem \cite{BCGJ23,GVS93,HV17,KV}.

We are not aware of any similar use of the weighted version of the problem, namely WTF2M, for the design of approximation algorithm, though this seems a natural direction for future work.

The Weighted $C_{\leq k}$-Free $2$-Matching problem (WCkF2M) is a natural generalization of WTF2M where we wish to compute a maximum-weight $2$-matching not inducing any cycle with up to $k$ edges. This problem is NP-hard for $k\geq 4$ even in bipartite graphs \cite{Kira}. As we already mentioned, WTF2M, i.e., the case $k=3$, is open. Under the assumption that the weight function satisfies a certain property, called \emph{vertex induced on every square}, 
Weighted $C_{\le 4}$-Free $2$-Matching problem in bipartite graphs can be solved in polynomial-time (see~\cite{Mak07,Tak09,PW21ESA}).

The situation is slightly different for the unweighted version CkF2M of WCkF2M. Papadimitriou showed that the case $k\geq 5$ is NP-hard (his proof is described by Cornu{\'{e}}jols and Pulleyblank in \cite{CP}). The complexity status of the problem is open for $k=4$. 
Hartvigsen~\cite{HARTVIGSEN2006} showed that the case where $k=4$ is polynomial-time solvable in bipartite graphs (see also \cite{Babenko12,Pap2007}).
B\'{e}rczi and V\'{e}gh~\cite{BV10} showed that the same problem can be solved in polynomial time if the input graph is sub-cubic. 

The $C_k$-Free 2-Matching problem is the problem of finding a maximum-cardinality 2-matching not inducing any cycle with exactly $k$ edges (while shorter cycles are allowed). In particular, TF2M is the special case of the above problem with $k=3$. 
B{\'e}rczi and Kobayashi~\cite{BERCZI2012565} showed that the square-free case, i.e., the case $k=4$, is polynomial-time solvable in sub-cubic graphs. Nam~\cite{Nam94} showed that the same problem can be solved in polynomial time if the squares are vertex-disjoint, and this result was generalized by Iwamasa et al.~\cite{IKT2024ESA}.

The problems get easier if we are allowed to take multiple copies of the same edge (i.e., the 2-matching is not simple). In particular, the Weighted Triangle-Free Non-Simple 2-Matching problem can be solved exactly in polynomial time \cite{CP}.

\subsection{Preliminaries}\label{sec:preliminaries}

We use standard graph notation. Let $G= (V, E, w)$ be an edge-weighted graph with $w\colon E\rightarrow \mathds{Q}_{\geq 0}$. 
Throughout this paper, we suppose that a graph may have self-loops but has no parallel edges.
Given a node $u\in V$ and a subset $F\subseteq E$ of the edges, 
the relative degree $d_F(u)$ of $u$ w.r.t.~$F$ is defined as the number of edges in $F$ incident to $v$, where self-loops are counted twice, that is, $d_F(u) = |\{uv\mid uv\in F\}| + 1$ if $uu\in F$ and $d_F(u) = |\{uv\mid uv\in F\}|$ otherwise.
An edge set $M \subseteq E$ is called a \textit{$2$-matching} if $d_M(v) \le 2$ for any $v \in V$. For a subset $F\subseteq E$ of the edges we denote $w(F) = \sum_{e\in F} w(e)$. 

A  \textit{trail $P$} (of length $|P|=k-1$) is a sequence of distinct consecutive edges $u_1u_2,u_2u_3,\ldots,u_{k-1}u_k$. The nodes $u_1,\ldots,u_k$ are not necessarily distinct. We say that $u_1$ and $u_k$ are the \textit{first} and the \textit{last} node of the trail, resp., and they are both \textit{endpoints} of the trail. Notice that we allow $u_1=u_k$, in which case we say that the trail is \emph{closed}.
For a trail $P$ of $G$, we denote by $\selfloop(P)$ the number of self-loops in $P$.
Given a trail $P$, we will also use $P$ to denote the corresponding set of edges $\{u_1u_2,\ldots,u_{k-1}u_k\}$ (neglecting their order); the meaning will be clear from the context.  
For two edge subsets $A_1, A_2 \subseteq E$, an \emph{alternating trail w.r.t.~$(A_1,A_2)$} is a (possibly closed) trail in $A_1\Delta A_2$ that alternates between the edges of $A_1 \setminus A_2$ and $A_2 \setminus A_1$, i.e., $u_iu_{i+1}\in A_1\setminus A_2$ if and only if $u_{i+1}u_{i+2}\in A_2\setminus A_1$ for $i=1,\dots,k-2$.

A \emph{triangle} is an edge set that forms a cycle of length $3$. An edge subset $F \subseteq E$ is called \emph{triangle-free} if it contains no triangle. Let $\mathcal{T}$ be a subset of triangles in $G=(V, E, w)$. 
An edge set $F\subseteq E$ is called $\calT$-\emph{free} if $F$ does not contain any triangle $T\in\mathcal{T}$ as a subset.
Note that when $\mathcal{T}$ is the set of all triangles in $G$, $\mathcal{T}$-freeness is equivalent to triangle-freeness. For a subset $\calT$ of the triangles of $G$ and a subset $F\subseteq E$ of the edges, we define $\calT(F)$ as the set of all triangles in $\calT$ containing at least one edge in $F$. The following observation is used multiple times throughout our argument.

\begin{observation}\label{obs:2edges}
    Let $G$ be a graph, let $\calT$ be a subset of the triangles of $G$, $M$ be a $2$-matching, and $T$ be a triangle of $G$ (possibly, $T\notin\calT$). If $|M\cap T|=2$ then $M$ is $\calT(T)$-free.
\end{observation}
\begin{proof}
    Let $T = \{u_1u_2, u_2u_3, u_1u_3\}$, and assume w.l.o.g.~that $T\cap M=\{u_1u_2,u_1u_3\}$. Assume to get a contradiction that $M$ is not $\calT(T)$-free, so that $M$ contains a triangle $T'\neq T$ that contains an edge of $T$, say w.l.o.g.~$u_1u_2$. By the degree constraint on $u_1$, the only other edge of $M$ incident to $u_1$ is $u_1u_3$. Then, since $T'\subseteq M$, $T'$ must also contain $u_1u_3$ and thus it is equal to $T$, a contradiction.
\end{proof}

\section{Weight Reduction and Proof of Theorem~\ref{thr:ptas}}
\label{sec:reduction}

As described in Section~\ref{sec:intro}, it is not obvious that Algorithm~\ref{alg:PTAS} runs in polynomial time. To resolve this issue, we reduce the problem to the case where the input instance has integer weights in $\{0,\ldots,\floor{n/\eps}\}$, where $\eps>0$ is some given constant. Here, $n$ denotes the number of vertices in the input graph.

\begin{lemma}\label{lem:red}
For any constants $\eps, \eps'\in (0, 1]$, given a polynomial-time $(1-\eps')$-approximation algorithm for WTF2M under the assumption that edge weights are integers in $\{0,\ldots,\floor{n/\eps}\}$, there is a polynomial-time $(1-\eps)(1-\eps')$-approximation algorithm for WTF2M (with no restrictions).    
\end{lemma}
\begin{proof}
Let $G=(V,E,w)$ be the input instance of WTF2M, and let $\OPT$ be some optimal solution for it.  
Let $e_{\max}$ be an edge with the maximum weight in $G$, and $W:=w(e_{\max})$. We may assume that $W > 0$, since otherwise all edges have weight zero and $w(\OPT)=0$. Note that $\{e_{\max}\}$ is a feasible solution of WTF2M, so $W\leq w(\OPT)$.
Next we define for each 
$e\in E$ the weight $w'(e)=\lfloor \frac{w(e)}{W}\frac{n}{\eps}\rfloor$. To the resulting instance 
$G'=(V,E,w')$ of WTF2M we apply the 
$(1-\eps')$-approximation algorithm with bounded integer weights that is given by the assumption of the lemma, and return the resulting solution $\APX$. Let $\OPT'$ be an optimal solution to $G'$. 
One has that
$$
w'(\OPT')\geq w'(\OPT) \geq \sum_{e \in \OPT} \left( \frac{w(e)}{W}\frac{n}{\eps} - 1 \right) \geq w(\OPT)\frac{n}{W\eps}-n\geq (1-\eps)w(\OPT)\frac{n}{W\eps}, 
$$
where we have used $|\OPT| \leq n$ and $W\leq w(\OPT)$ in the third and the last inequalities, respectively.
Hence
$$
w(\APX)\geq \frac{W\eps}{n}w'(\APX)\geq \frac{W\eps(1-\eps')}{n}w'(\OPT')\geq (1-\eps)(1-\eps')w(\OPT),  
$$
which completes the proof. 
\end{proof}

By combining Theorem~\ref{thm:main} and Lemma~\ref{lem:red}, we obtain a PTAS for WTF2M and prove Theorem~\ref{thr:ptas} as follows.

\begin{proof}[Proof of Theorem \ref{thr:ptas}]
    By Lemma~\ref{lem:red} it is sufficient to show that Algorithm~\ref{alg:PTAS} is a PTAS for the problem with bounded integer weights. The feasibility of the solution is obvious by construction. Clearly, every iteration of the while loop can be executed in $n^{O(1/\eps)}$ time. Each time this loop is executed, the weight of the solution grows by at least $1$ (since all weights are integers). Since the initial solution has weight $0$ and no solution can have weight more than $n^2/\eps$ (corresponding to $n$ edges of maximum possible weight), the number of iterations of the while loop is bounded by $n^2/\eps$. Then, since in each iteration a trail can be found in $n^{O(1/\eps)}$ time, the overall running time is also bounded by $n^{O(1/\eps)}$.
    By Theorem~\ref{thm:main}, as long as $(1-\eps)\opt(G) > w(\APX)$, the value of $\APX$ is updated. This implies that, when the algorithm terminates, $(1-\eps)\opt(G) \leq w(\APX)$.
    \end{proof}

What remains is the proof of Theorem~\ref{thm:main}. In what follows in this paper, we prove the following stronger theorem.

\begin{theorem}\label{thm:strongerMain}
    Let $G = (V, E, w)$ be a weighted graph. Let $\calT$ be a subset of the triangles of $G$, let $A_1, A_2$ be two $\calT$-free $2$-matchings of $G$, and let $\eps\in(0, 1]$ be a constant. If $(1-\eps)w(A_1) > w(A_2)$, then there exists an alternating trail $P$ w.r.t.~$(A_1, A_2)$ such that: \begin{enumerate*}[label=(\arabic*)]
        \item $A_2\Delta P$ is a $\calT$-free $2$-matching,
        \item $w(A_2\Delta P) > w(A_2)$, and
        \item $|P|+2\selfloop(P)\leq 7/\eps$.
    \end{enumerate*} 
\end{theorem}

Recall that $G$ may have self-loops but has no parallel edges.
Notice that we can derive Theorem~\ref{thm:main} from Theorem~\ref{thm:strongerMain} by setting $A_1 = \OPT$ and $A_2 = \APX$.
To prove Theorem~\ref{thm:strongerMain}, we first consider the case of $\calT = \emptyset$ in Section~\ref{sec:basecase}, and then consider the general case in Section~\ref{sec:mainproof}. 

\section{Case with No Forbidden Triangles}
\label{sec:basecase}

In this section, we prove a special case of Theorem~\ref{thm:strongerMain} where $\calT = \emptyset$.

\begin{lemma}\label{lem:baseCase}
    Let $G = (V, E, w)$ be an edge-weighted graph, let $A_1, A_2$ be two 2-matchings of $G$, and let $\eps\in(0, 1]$ be a constant. If $(1-\eps)w(A_1) > w(A_2)$, then there exists an alternating trail $P$ w.r.t. $(A_1, A_2)$ such that: \begin{enumerate*}[label=(\arabic*)]
        \item $A_2\Delta P$ is a 2-matching,
        \item $w(A_2\Delta P) > w(A_2)$, and
        \item $|P|\leq 3/\eps$.
    \end{enumerate*} 
\end{lemma}
\begin{proof}
    Let $\mathcal{P}$ be a partition of $A_1\Delta A_2$ into alternating trails w.r.t.\ $(A_1, A_2)$ such that $A_2\Delta P$ is a 2-matching for every $P\in \mathcal{P}$. The existence of such a partition is a well-known fact, see e.g.~\cite[Lemma 2.1]{KN2025}. It holds that
    
    \begin{align*}
        \sum_{P\in\calP}w(A_2\cap P) &= w(A_2) -w(A_1\cap A_2)\\
        &<(1-\eps)w(A_1)-w(A_1\cap A_2)\\
        &\leq (1-\eps)(w(A_1)-w(A_1\cap A_2))\\
        &=\sum_{P\in\calP} (1-\eps) w(A_1\cap P).
    \end{align*}
    Therefore there exists $P^*\in\calP$ satisfying that $w(A_2\cap P^*)<(1-\eps)w(A_1\cap P^*)$.
    
    Let $m:=\big\lceil \frac{1}{\varepsilon} \big\rceil$. Since $w(A_1\cap P^*) > w(A_2\cap P^*)$, one has that $w(A_2\Delta P^*) - w(A_2) > 0$, and thus $P^*$ satisfies the claim of the lemma if $|P^*|\le 2m-1\le 3/\eps$.
    Now we assume that $P^* = u_1 u_2, u_2 u_3, \dots, u_k u_{k+1}$ satisfies $|P^*| > 2m-1$.
    We consider the following two cases separately. 
    
    Suppose that $P^*$ is closed and even, i.e., $u_1 = u_{k+1}$, and one of $u_1 u_2$ and $u_k u_{k+1}$ is in $A_1$ while the other is in $A_2$. Let $\calS$ be the set of alternating subtrails $P$ of $P^*$ with $|P|=2m-1$ whose first and last edges are in $A_2$ (including subtrails passing through $u_1 = u_{k+1}$). That is, each trail in $\calS$ is of the form $u_i u_{i+1}, u_{i+1} u_{i+2}, \dots, u_{i+2m-2} u_{i+2m-1}$, where the indices are taken modulo $k$.
    It is easy to see that $A_2\Delta P$ is a 2-matching for any $P\in \calS$, because $d_{A_2\Delta P}(v) \le d_{A_2}(v) \le 2$ for any $v \in V$.
    Every edge in $A_1\cap P^*$ is contained in exactly $m-1$ trails in $\calS$, while every edge in $A_2\cap P^*$ is contained in exactly $m$ trails in $\calS$.
    Therefore,
        \begin{align*}
            \sum_{P\in\calS}\left( w(A_2\Delta P)-w(A_2) \right) &= \sum_{P\in\calS} \left( w(A_1\cap P)-w(A_2\cap P) \right) \\
            &=(m-1)w(A_1\cap P^*)-m\cdot w(A_2\cap P^*)\\
            &>m\bigg(\bigg(1-\frac{1}{m}\bigg)-(1-\eps)\bigg)w(A_1\cap P^*)\geq 0.
        \end{align*}
    In the last inequality we have used that $m=\ceil{\frac{1}{\eps}}$ and in the second-last one the fact that $w(A_2\cap P^*)<(1-\eps)w(A_1\cap P^*)$. This implies the existence of $P\in\calS$ such that $w(A_2\Delta P)-w(A_2)>0$, which satisfies the claim of the lemma.

    Suppose now that $P^*$ is not closed or not even, i.e., $u_1\neq u_{k+1}$, or $u_1=u_{k+1}$ and both $u_1u_2$ and $u_ku_{k+1}$ are in $A_i$ for some $i\in\{1, 2\}$.
    Let $\calS$ be the set of alternating subtrails $P$ of $P^*$ satisfying one of the following conditions:
        \begin{itemize}
            \item $|P|=2m-1$, and both the first and last edges of $P$ are contained in $A_2$; or
            \item $|P|<2m-1$,
            where either the first or last edge of $P$ coincides the first or last edge of $P^*$, and the other extremal edge of $P$ is contained in $A_2$ (possibly, $|P|=1$).
        \end{itemize}
    Note that each trail in $\calS$ is of the form $u_i u_{i+1}, u_{i+1} u_{i+2}, \dots, u_{i+\ell-1} u_{i+\ell}$ for some $\ell \le 2m-1$, where the indices are {\it not} taken modulo $k$.
    One sees that $A_2\Delta P$ is a 2-matching for every $P\in \calS$, because $d_{A_2\Delta P}(v) \le \max\{ d_{A_2}(v), d_{A_2 \Delta P^*}(v) \}\le 2$ for every $v \in V$.
    Every edge in $A_1\cap P^*$ is contained in exactly $m-1$ trails in $\calS$, while every edge in $A_2\cap P^*$ is contained in exactly $m$ trails in $\calS$. Then, by the same argument as in the previous case, we see that there exists $P\in\calS$ such that $w(A_2\Delta P)-w(A_2)>0$, which satisfies the claim of the lemma.
\end{proof}

This lemma together with the following observation, proves the special case of Theorem~\ref{thm:strongerMain} in which $\calT = \emptyset$.

\begin{observation}\label{cor:lengthsl}
    For a trail $P$ satisfying the condition in Lemma~\ref{lem:baseCase}, it holds that $|P|+2\selfloop(P)\le \frac{7}{\eps}$. 
\end{observation}
\begin{proof}
    Since $G$ has no parallel edges, $\selfloop(P)$ is at most the number of nodes of $P$, which is equal to $|P|-\selfloop(P)+1$.
    Therefore, $\selfloop(P)\le\frac{|P|+1}{2}$, so $|P|+2\selfloop(P)\le 2|P|+1\le \frac{7}{\eps}$.
\end{proof}

\section{Proof of Theorem~\ref{thm:strongerMain}}\label{sec:mainproof}
    We prove Theorem~\ref{thm:strongerMain} by induction on $|\calT|$. The base case, i.e., when $\calT=\emptyset$, is given by Lemma~\ref{lem:baseCase} and Observation~\ref{cor:lengthsl}. Suppose that $|\calT|>0$ holds. 
    We consider the following cases, assuming at each case that the previous ones do not hold:
    \begin{description}
        \item[(1)] There exists a triangle $T\in\calT$ not contained in $A_1\cup A_2$.
        \item[(2)] There exists a triangle $T \in \calT$ such that $|A_1 \cap T| = |A_2 \cap T| = 2$. 
        \item[(3)] There exists a triangle $T \in \calT$ such that $|A_1 \cap T| = 2$ and $A_1 \Delta T$ is $\calT$-free. 
        \item[(4)] There exists a triangle $T \in \calT$ such that $|A_2 \cap T| = 2$ and $A_2 \Delta T$ is $\calT$-free. 
        \item[(5)] There exists a triangle $T \in \calT$ such that $|A_1 \cap T| = 2$. 
        \item[(6)] There exists a triangle $T \in \calT$ such that $|A_2 \cap T| = 2$. 
    \end{description}
    For each $T \in \calT$ with $T \subseteq A_1 \cup A_2$, it holds that $|A_1 \cap T| = 2$ or $|A_2 \cap T| = 2$, because $A_1$ and $A_2$ are $\calT$-free.  
    Thus, Cases~\ref{case:trivial},~\ref{case:final1}, and~\ref{case:final2} cover all possible cases, and thus Theorem~\ref{thm:strongerMain} follows from these six cases. Notice that the fact that Cases~\ref{case:A1A2},~\ref{case:2A1A2}, and~\ref{case:A12A2} do not hold is necessary for the proof of Cases~\ref{case:final1} and~\ref{case:final2}. 
    Each case is discussed in the following. In the figures that we will use to illustrate the different cases, solid lines represent edges in $A_1$, while dashed lines represent those in $A_2$.
    
\bigskip

    \begin{caseanalysis}
    \case{There exists a triangle $T\in\calT$ not contained in $A_1\cup A_2$.}\label{case:trivial}
    
    Define $\calT' := \calT\setminus\{T\}$. Since $A_1$ and $A_2$ are both $\calT'$-free $2$-matchings and $|\calT'| < |\calT|$, by the induction hypothesis, there exists an alternating trail $P$ w.r.t.~($A_1$, $A_2$) such that $|P| + 2\selfloop(P)\leq 7/\eps$, $w(A_2\Delta P) > w(A_2)$, and $A_2\Delta P$ is a $\calT'$-free $2$-matching of $G$. Since $T$ is not contained in $A_1\cup A_2$, $A_2\Delta P$ does not contain $T$, hence $A_2\Delta P$ is $\calT$-free. Altogether $P$ satisfies the claim of the theorem.

\bigskip

    \case{There exists a triangle $T\in\calT$ such that $|A_1 \cap T| = |A_2 \cap T| = 2$.}\label{case:A1A2} 

    Let $T=\{u_1u_2, u_2u_3, u_1u_3\}$.
    Since both $A_1$ and $A_2$ are $\calT$-free and Case~\ref{case:trivial} does not hold, we can assume w.l.o.g. that $u_2 u_3 \in A_1 \cap A_2$, $u_1u_2\in A_1\setminus A_2$, and $u_1u_3\in A_2\setminus A_1$ (see Figure~\ref{fig:A1A2} left). We define an auxiliary weighted graph $G' = (V', E', w')$ as follows:
    \begin{itemize}
        \item $V':= V\cup\{u_1'\}$, $E':= (E \setminus\{u_1u_2, u_1u_3\}) \cup\{u_1'u_2, u_1'u_3, u_1u_1'\}$,
        \item $w'(e) = w(e)$ for all $e\in E\setminus\{u_1u_2, u_1u_3\}$,
        \item $w'(u_1'u_2) = w(u_1u_2)$, $w'(u_1'u_3) = w(u_1u_3)$, $w'(u_1u_1') = 0$.  
    \end{itemize}
We also define:
    \begin{itemize}
        \item $A_1':=(A_1\setminus\{u_1u_2\})\cup\{u_1'u_2, u_1u_1'\}$, $A_2':=(A_2\setminus\{u_1u_3\})\cup\{u_1'u_3, u_1u_1'\}$,
        \item $\calT'=\calT\setminus\calT(\{u_1u_2, u_1u_3\})$.
    \end{itemize}
    
    \begin{figure}
        \centering
        \includegraphics{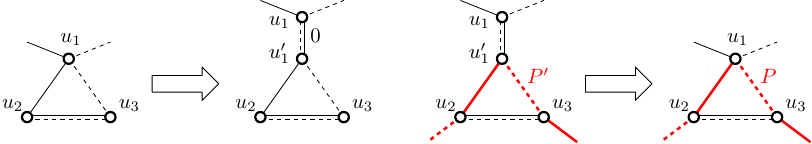}
        \caption{Reduction for Case~\ref{case:A1A2}.} \label{fig:A1A2}
    \end{figure}
    
    Observe that $(1-\eps)w'(A_1')=(1-\eps)w(A_1)>w(A_2)=w'(A_2')$.
    We claim that $A_i'$ is a $\calT'$-free $2$-matching of $G'$, for $i\in\{1, 2\}$. Since $d_{A_i'}(u) = d_{A_i}(u)\leq 2$ for every node $u\in V$, and $d_{A_i'}(u_1') = 2$, $A_i'$ is a $2$-matching of $G'$. We also observe that the only edges in $A_i'$ that are not present in $A_i$ are incident to $u_1'\notin V$, so they are not edges of triangles in $\calT$, and thus they are not edges of triangles in $\calT'$. This, together with the fact that $A_i$ is $\calT$-free implies that $A_i'$ is $\calT'$-free, for $i\in\{1, 2\}$. Thus, $A_i'$ is a $\calT'$-free $2$-matching of $G'$, for $i\in\{1, 2\}$. Also, since $T$ contains both $u_1u_2$ and $u_1u_3$, one has $|\calT'| < |\calT|$. One can apply the induction hypothesis to find an alternating trail $P'$ w.r.t.~$(A_1', A_2')$ such that $|P'| + 2\selfloop(P')\leq 7/\eps$, $w'(A_2'\Delta P') > w'(A_2')$, and $A_2'\Delta P'$ is a $\calT'$-free $2$-matching of $G'$. 

    Let $P$ be the trail in $G$ resulting from replacing $u_1'u_2$ by $u_1u_2$ and $u_1'u_3$ by $u_1u_3$ in $P'$, if any of them is present in $P'$ (see Figure~\ref{fig:A1A2} right). Notice that $P$ is an alternating trail w.r.t.~$(A_1, A_2)$. We claim that $P$ satisfies the conditions of the theorem.
    
    Notice that $|P| + 2\selfloop(P) = |P'| + 2\selfloop(P')\leq 7/\eps$. We now prove that $w(A_2\Delta P) > w(A_2)$. 
    Observe that $A_2$ (resp.,~$A_2\Delta P$) in $G$ is obtained from $A'_2$ (resp.,~$A'_2 \Delta P'$) in $G'$ by contracting $u_1 u'_1$ (i.e., by removing $u_1 u'_1$ and by identifying $u'_1$ with $u_1$). Since $w'(u_1'u_2)=w(u_1u_2)$, $w'(u_1'u_3)=w(u_1u_3)$, and $w'(u_1u_1')=0$, this observation shows that $w(A_2\Delta P) = w'(A_2'\Delta P') > w'(A_2') = w(A_2)$. 
    
    Let us now show that $A_2\Delta P$ is a $2$-matching. Since $A_2\Delta P$ is obtained from $A'_2 \Delta P'$ by contracting $u_1 u'_1$, for every node $u\in V\setminus\{u_1\}$, $d_{A_2\Delta P}(u) = d_{A_2'\Delta P'}(u)\leq 2$. 
    Since the edge $u_1u_1'$ counts toward both $d_{A_2'\Delta P'}(u_1)$ and $d_{A_2'\Delta P'}(u_1')$, but is not present in $A_2\Delta P$, one has that $d_{A_2\Delta P}(u_1) = d_{A_2'\Delta P'}(u_1) - 1 + d_{A_2'\Delta P'}(u_1') - 1 \leq 2$. 
    
    Finally, it is left to show that $A_2\Delta P$ is $\calT$-free. Since $A_2'\Delta P'$ is $\calT'$-free, $A_2\Delta P$ is $\calT'$-free. Now, notice that the edge $u_1'u_2$ belongs to $P'$ only if $u_1'u_3$ also belongs to $P'$ because $d_{A_2'\Delta P'}(u_1')\leq 2$. Thus, if $u_1'u_2\in P'$, both $u_1u_2, u_1u_3\in P$ and $(A_2\Delta P)\cap T=\{u_1u_2, u_2u_3\}$, so $A_2\Delta P$ is $\calT(T)$-free by Observation~\ref{obs:2edges} (implying that it is $\calT(\{u_1u_2, u_1u_3\})$-free). Assume now $u_1'u_2\notin P'$, so $u_1u_2\notin P$ and thus $u_1u_2\notin A_2\Delta P$. If $u_1u_3\in A_2\Delta P$ then $(A_2\Delta P)\cap T=\{u_1u_3, u_2u_3\}$, so $A_2\Delta P$ is $\calT(T)$-free by Observation~\ref{obs:2edges} (implying that it is $\calT(\{u_1u_2, u_1u_3\})$-free). On the other hand, if $u_1u_3\notin A_2\Delta P$, then one has that $u_1u_2, u_1u_3\notin A_2\Delta P$, and thus $A_2\Delta P$ is $\calT(\{u_1u_2, u_1u_3\})$-free. The claim follows.

\bigskip

    \case{There exists a triangle $T\in \calT$ such that $|A_1 \cap T| = 2$ and $A_1 \Delta T$ is $\calT$-free.}\label{case:2A1A2}
    
    Let $T=\{u_1u_2, u_2u_3, u_1u_3\}\in \calT$.
    Since Cases~\ref{case:trivial} and~\ref{case:A1A2} do not hold, we can assume w.l.o.g.~that $\{u_1u_2, u_1u_3\}\subseteq A_1\setminus A_2$ and $u_2u_3 \in A_2 \setminus A_1$. 
    We consider the following subcases. 
    
    \subcase{$u_1u_1\in A_2$.}\label{case:2A1A2loop} 
    If $w(u_1u_2)+w(u_1u_3) > w(u_2u_3)+w(u_1u_1)$ then let $P:=\{u_1u_2, u_2u_3, u_1u_3, u_1u_1\}$. Let us show that $P$ satisfies the claim of the theorem. Indeed, since $w(u_1u_2)+w(u_1u_3) > w(u_2u_3)+w(u_1u_1)$, $w(A_2\Delta P) > w(A_2)$. One has $|P| + 2\selfloop(P) = 6 < 7/\eps$, and since $P$ is an even closed alternating trail, $A_2\Delta P$ is a $2$-matching. Moreover, since $(A_2\Delta P)\cap T=\{u_1u_2, u_1u_3\}$, $A_2\Delta P$ is $\calT(T)$-free by Observation~\ref{obs:2edges}. Since the only edge in $P$ not in $T$ is $u_1u_1$, this implies that $A_2\Delta P$ is $\calT$-free. 
    
    From now on we assume that $w(u_1u_2)+w(u_1u_3) \leq w(u_2u_3)+w(u_1u_1)$. We set (see Figure~\ref{fig:2A1A2loop}):

    \begin{itemize}
        \item $A_1':= (A_1\setminus\{u_1u_2, u_1u_3\})\cup \{u_2u_3, u_1u_1\}$, $A_2':= A_2$,
        \item $w'(e):=w(e)$ for all $e\in E$,
        \item $\calT':=\calT\setminus\calT(\{u_1u_2, u_1u_3\})$.
    \end{itemize}

    Notice that $A_1'$ and $A_2'$ are both $2$-matchings of $G$. Clearly, $A_2'$ is $\calT'$-free, and  
    $A_1 \Delta T$ is $\calT$-free by this case assumption, so $A_1' = (A_1 \Delta T) \cup \{u_1 u_1\}$ is also $\calT'$-free. Thus, $A_i'$ is a $\calT'$-free $2$-matching of $G$, for $i\in\{1, 2\}$. Also, $w'(A_2') = w(A_2)$ and $w'(A_1')= w(A_1) + w(u_2u_3)+w(u_1u_1) - w(u_1u_2) - w(u_1u_3)\geq w(A_1)$, because $w(u_1u_2)+w(u_1u_3) \leq w(u_2u_3)+w(u_1u_1)$. By the assumption of the theorem, one has $(1-\eps)w'(A_1')\geq (1-\eps)w(A_1) > w(A_2) = w'(A'_2)$.
    
    \begin{figure}
        \centering
        \includegraphics{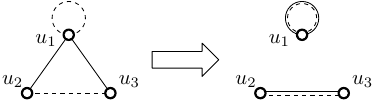}
        \caption{Reduction for Case~\ref{case:2A1A2loop}}
        \label{fig:2A1A2loop}
    \end{figure}
    
    Since $A_1'$ and $A_2'$ are both $\calT'$-free $2$-matchings and $(1-\eps)w'(A_1') > w(A_2')$, one can apply the induction hypothesis to find an alternating trail $P'$ w.r.t.~$(A_1', A_2')$ such that $|P'| + 2\selfloop(P')\leq 7/\eps$, $w'(A_2'\Delta P')>w'(A_2')$, and $A_2'\Delta P'$ is a $\calT'$-free $2$-matching of $G'$. Let $P:=P'$, and notice that $P$ is an alternating trail w.r.t.~$(A_1, A_2)$. We claim that $P$ satisfies the conditions of the theorem. Observe that $|P| + 2\selfloop(P) = |P'| + 2\selfloop(P')\leq 7/\eps$ and $w(A_2\Delta P)=w'(A_2'\Delta P')>w'(A_2')=w(A_2)$.

    Since $A_2'=A_2$ and $P'=P$, it holds that $d_{A_2\Delta P}(u) = d_{A_2'\Delta P'}(u)\leq 2$ for every $u\in V$, so $A_2\Delta P$ is a $2$-matching. Moreover, since $A_2'\Delta P'$ is $\calT'$-free, so is $A_2\Delta P$. Also, since $u_1u_2, u_1u_3\notin A_1'\cup A_2'$, one has $u_1u_2, u_1u_3\notin P$, and thus $u_1u_2, u_1u_3\notin A_2\Delta P$. Therefore, $A_2\Delta P$ is $\calT(\{u_1u_2, u_1u_3\})$-free. The claim follows.

    \subcase{$u_1u_1\notin A_2$.}\label{case:2A1A2:red} 
    Notice that $u_1u_1\notin A_1$, because of the degree constraints on $u_1$.
We define
    \begin{itemize}
        \item $A_2':= A_2$,
        \item $w'(e):=w(e)$ for all $e\in E\setminus\{u_1u_1\}$,
        \item $\calT':=\calT\setminus\calT(\{u_1u_2, u_1u_3\})$.
    \end{itemize}    
If $w(u_1u_2)+w(u_1u_3)-w(u_2u_3)\geq 0$, in case $u_1u_1\notin E$, we add $u_1u_1$ to $E$ with some arbitrary non-negative weight. Furthermore, we set
    \begin{itemize}
        \item $A_1':= (A_1\setminus\{u_1u_2, u_1u_3\})\cup \{u_2u_3, u_1u_1\}$,
  \item $w'(u_1u_1):=w(u_1u_2)+w(u_1u_3)-w(u_2u_3)$.
    \end{itemize}
Otherwise, if $u_1u_1\in E$, we remove it from $E$. Furthermore, we set
    \begin{itemize}
        \item $A_1':= (A_1\setminus\{u_1u_2, u_1u_3\})\cup \{u_2u_3\}$.
    \end{itemize}
    
    Notice that $A_1'$ and $A_2'$ are both $2$-matchings of $G$. Clearly, $A_2'$ is $\calT'$-free, and 
    $A_1 \Delta T$ is $\calT$-free by this case assumption, so $A_1' \subseteq (A_1 \Delta T) \cup \{u_1 u_1\}$ is also $\calT'$-free. 
    Note also that $w'(A_2')=w(A_2)$ and $w'(A_1') \geq w(A_1)$.
        
    \begin{figure}
        \centering
        \includegraphics{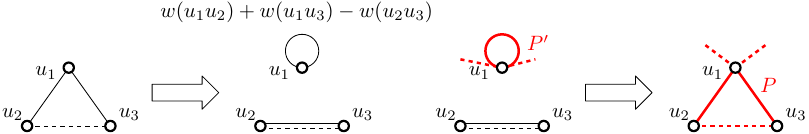}
        \caption{Reduction for Case~\ref{case:2A1A2:red}.}
        \label{fig:2A1A2}
    \end{figure}

    Since $A_1'$ and $A_2'$ are both $\calT'$-free $2$-matchings and $w'(A_2')=w(A_2) <(1-\eps)w(A_1)\le (1-\eps)w'(A_1')$, one can apply the induction hypothesis to find an alternating trail $P'$ w.r.t.~$(A_1', A_2')$ such that $|P'| + 2\selfloop(P')\leq 7/\eps$, $w'(A_2'\Delta P')>w'(A_2')$, and $A_2'\Delta P'$ is a $\calT'$-free $2$-matching of $G'$. If $u_1u_1\in P'$, then let $P$ be the alternating trail resulting from inserting the trail $u_1u_2, u_2u_3, u_1u_3$ in place of $u_1u_1$ (see Figure~\ref{fig:2A1A2} right). If $u_1u_1\notin P'$, then let $P:=P'$. Notice that $P$ is an alternating trail w.r.t.~$(A_1, A_2)$. We claim that $P$ satisfies the conditions of the theorem, i.e., $A_2\Delta P$ is a $\calT$-free $2$-matching, $|P| + 2\selfloop(P)\leq 7/\eps$, and $w(A_2\Delta P)>w(A_2)$.

    Notice that $|P| + 2\selfloop(P)=|P'| + 2\selfloop(P')\leq 7/\eps$ holds regardless of whether $u_1u_1 \in P'$ or not.
    If $u_1u_1\notin P'$, since $A_2=A_2'$ and $P=P'$, we have $w(A_2\Delta P)=w'(A_2'\Delta P')>w'(A_2')=w(A_2)$. On the other hand, if $u_1u_1\in P'$, then $A_2\Delta P=((A_2'\Delta P')\setminus\{u_1u_1, u_2u_3\})\cup \{u_1u_2, u_1u_3\}$. Therefore, 
    \begin{align*}
    w(A_2\Delta P)&=w'(A_2'\Delta P')-w'(u_1u_1) - w'(u_2u_3) + w(u_1u_2) + w(u_1u_3) \\
                  &= w'(A_2'\Delta P')>w'(A_2')=w(A_2), 
    \end{align*}
    where in the second equality we have used that $w'(u_1u_1)=w(u_1u_2)+w(u_1u_3)-w(u_2u_3)$ and $w'(u_2u_3) = w(u_2u_3)$.
    In both cases, it holds that $w(A_2\Delta P)>w(A_2)$. 
    
    It is left to show that $A_2\Delta P$ is a $\calT$-free $2$-matching. Recall that $A_2'=A_2$. If $u_1u_1\in P'$, then $P'=P$ and thus $d_{A_2\Delta P}(u)=d_{A_2'\Delta P'}(u)\leq 2$ for all $u\in V$. On the other hand, if $u_1u_1\notin P'$, then $A_2\Delta P=((A_2'\Delta P')\setminus\{u_1u_1, u_2u_3\})\cup \{u_1u_2, u_1u_3\}$. Therefore, $d_{A_2\Delta P}(u)=d_{A_2'\Delta P'}(u)\leq 2$ for all $u\in V$. Thus, $A_2\Delta P$ is a $2$-matching. 
    
    Since $A_2'\Delta P'$ is $\calT'$-free, $A_2\Delta P\subseteq (A_2'\Delta P')\cup \{u_1u_2, u_1u_3\}$, and neither $u_1u_2$ nor $u_1u_3$ is contained in a triangle in $\calT'$, one has that $A_2\Delta P$ is $\calT'$-free. Now, if $u_1u_1\in P'$, then $u_1u_2, u_1u_3, u_2u_3\in P$ and thus $(A_2\Delta P)\cap T = \{u_1u_2, u_1u_3\}$, so $A_2\Delta P$ is $\calT(\{u_1u_2, u_1u_3\})$-free by Observation~\ref{obs:2edges}. If $u_1u_1\notin P'$, then since $P=P'$ contains neither $u_1u_2$ nor $u_1u_3$,  $A_2\Delta P$ is also $\calT(\{u_1u_2, u_1u_3\})$-free. Therefore, $A_2\Delta P$ is $\calT$-free. The claim follows.

\bigskip
    \case{There exists a triangle $T\in \calT$ such that $|A_2 \cap T| = 2$ and $A_2 \Delta T$ is $\calT$-free.}\label{case:A12A2} 

    Let $T=\{u_1u_2, u_2u_3, u_1u_3\}$.
    Since Cases~\ref{case:trivial} and~\ref{case:A1A2} do not hold, we can assume w.l.o.g.~that $\{u_1u_2, u_1u_3\}\subseteq A_2\setminus A_1$ and $u_2u_3 \in A_1 \setminus A_2$. 
    We consider the following subcases.

    \subcase{$u_1u_1\in A_1$.}\label{case:A12A2loop} 
    If $w(u_1u_1) + w(u_2u_3) > w(u_1u_2) + w(u_1u_3)$ then let $P:=\{u_1u_2, u_2u_3, u_1u_3, u_1u_1\}$. Let us show that $P$ satisfies the claim of the theorem. Indeed, since $w(u_1u_1) + w(u_2u_3) > w(u_1u_2) + w(u_1u_3)$, $w(A_2\Delta P) > w(A_2)$. One has $|P| + 2\selfloop(P) = 6 < 7/\eps$, and since $P$ is an even closed alternating trail, $A_2\Delta P$ is a $2$-matching. Moreover,  
    $A_2 \Delta T$ is $\calT$-free by this case assumption, so $A_2 \Delta P = (A_2 \Delta T) \cup \{u_1 u_1\}$ is also $\calT$-free.

    From now on we assume that $w(u_1u_1) + w(u_2u_3) \leq w(u_1u_2) + w(u_1u_3)$. We set (see Figure~\ref{fig:A12A2loop}):

    \begin{itemize}
        \item $A_1'=A_1$, $A_2':= (A_2\setminus\{u_1u_2, u_1u_3\})\cup \{u_2u_3, u_1u_1\}$,
        \item $w'(e):=w(e)$ for all $e\in E\setminus\{u_1u_1, u_2u_3\}$,
        \item $w'(u_1u_1)=w(u_1u_2)+w(u_1u_3)$, $w'(u_2u_3)=0$,
        \item $\calT':=\calT\setminus\calT(\{u_1u_2, u_1u_3\})$.
    \end{itemize}
    
    Notice that $A_1'$ and $A_2'$ are $2$-matchings of $G$. Clearly, $A_1'$ is $\calT'$-free, and $A_2 \Delta T$ is $\calT$-free by this case assumption, so $A_2' = (A_2 \Delta T) \cup \{u_1 u_1\}$ is also $\calT'$-free. 
    Note as well that $w'(A_2')=w(A_2)$ and 
    \begin{align*}
        w'(A_1') &= w(A_1) - w(u_1u_1) - w(u_2u_3) + w'(u_1u_1) + w'(u_2u_3)  \\
                 &= w(A_1) - w(u_1u_1) - w(u_2u_3) + w(u_1u_2) + w(u_1u_3)  \geq w(A_1), 
    \end{align*}
    where in the final inequality we used $w(u_1u_1) + w(u_2u_3) \leq w(u_1u_2) + w(u_1u_3)$.
    Thus, by the assumption of the theorem, $w'(A_2') = w(A_2) < (1-\eps)w(A_1)\leq (1-\eps)w'(A_1')$.
    
    \begin{figure}
        \centering
        \includegraphics{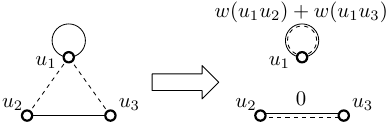}
        \caption{Reduction for Case~\ref{case:A12A2loop}}
        \label{fig:A12A2loop}
    \end{figure}
    
    Since $A_1'$ and $A_2'$ are both $\calT'$-free $2$-matchings and $w'(A_2') < (1-\eps)w'(A_1')$, one can apply the induction hypothesis to find an alternating trail $P'$ w.r.t.~$(A_1', A_2')$ such that $|P'| + 2\selfloop(P')\leq 7/\eps$, $w'(A_2'\Delta P)>w(A_2')$, and $A_2'\Delta P'$ is a $\calT'$-free $2$-matching of $G'$. Let $P:=P'$, and notice that $P$ is an alternating trail w.r.t.~$(A_1, A_2)$. We claim that $P$ satisfies the conditions of the theorem. 
    
    One has $|P| + 2\selfloop(P) = |P'| + 2\selfloop(P')\leq 7/\eps$. Notice that $A_2\Delta P=((A_2'\Delta P')\setminus\{u_1u_1, u_2u_3\})\cup \{u_1u_2, u_1u_3\}$. Thus, 
    \begin{align*}
        w(A_2\Delta P)&=w'(A_2'\Delta P') - w'(u_1u_1) - w'(u_2u_3) + w(u_1u_2) + w(u_1u_3) \\
                      &= w'(A_2'\Delta P') > w'(A_2')=w(A_2),
    \end{align*}
    where in the second equality we used $w'(u_1u_1)= w(u_1u_2) + w(u_1u_3)$ and $w'(u_2u_3)=0$.
    
    It is left to show that $A_2\Delta P$ is a $\calT$-free $2$-matching. Since $A_2\Delta P=((A_2'\Delta P')\setminus\{u_1u_1, u_2u_3\})\cup \{u_1u_2, u_1u_3\}$, it holds that $d_{A_2\Delta P}(u) = d_{A_2'\Delta P'}(u)\leq 2$ for every $u\in V$, implying $A_2\Delta P$ is a $2$-matching. 
    
    Since $A_2'\Delta P'$ is $\calT'$-free, $A_2\Delta P\subseteq (A_2'\Delta P')\cup\{u_1u_2, u_1u_3\}$, and neither $u_1u_2$ nor $u_1u_3$ is contained in a triangle in $\calT'$, one has that $A_2\Delta P$ is $\calT'$-free. Also, since $A_2\Delta P = ((A'_2\Delta P')\setminus\{u_1u_1, u_2u_3\})\cup\{u_1u_2, u_1u_3\}$, one sees that $(A_2\Delta P)\cap T=\{u_1u_2, u_1u_3\}$, so $A_2\Delta P$ is $\calT(\{u_1u_2, u_1u_3\})$-free by Observation~\ref{obs:2edges}, implying that it is $\calT$-free. The claim follows.

    \subcase{$u_1u_1\notin A_1$.}\label{case:A12A2:red} 
        If $w(u_2u_3) > w(u_1u_2) + w(u_1u_3)$ then let $P:=\{u_1u_2, u_2u_3, u_1u_3\}$. Let us show that $P$ satisfies the claim of the theorem. Indeed, since $w(u_2u_3) > w(u_1u_2) + w(u_1u_3)$, $w(A_2\Delta P) > w(A_2)$. One has $|P| + 2\selfloop(P) = 3 < 7/\eps$, and since $d_{A_2 \Delta P}(v) \le d_{A_2}(v) \le 2$ for every $v \in V$, $A_2\Delta P$ is a $2$-matching. Moreover,  
        $A_2 \Delta P = A_2 \Delta T$ is $\calT$-free by this case assumption.

        Next we assume that $w(u_2u_3) \le w(u_1u_2) + w(u_1u_3)$.
        Notice that $u_1u_1\notin A_2$, because of the degree constraint on $u_1$. We assume that $u_1u_1\in E$, otherwise we add it to $E$ with an arbitrary non-negative weight. 
        Let (see Figure~\ref{fig:A12A2}):

    \begin{itemize}
        \item $A_1':=A_1$, $A_2':=(A_2\setminus\{u_1u_2, u_1u_3\})\cup\{u_1u_1, u_2u_3\}$,
        \item $w'(e):=w(e)$ for all $e\in E\setminus\{u_1u_1\}$,
        \item $w'(u_1u_1) := w(u_1u_2) + w(u_1u_3) - w(u_2u_3)$,
        \item $\calT':=\calT\setminus\calT(\{u_1u_2, u_1u_3\})$.
    \end{itemize}
    
    Notice that $A_1'$ and $A_2'$ are $2$-matchings of $G$. Clearly, $A_1'$ is $\calT'$-free, and 
    $A_2 \Delta T$ is $\calT$-free by this case assumption, so $A_2' = (A_2 \Delta T) \cup \{u_1 u_1\}$ is also $\calT'$-free. 
    Notice that $w'(A_1')=w(A_1)$ and $w'(A_2')=w(A_2)$. 
    We also note that $w'(u_1u_1) \geq 0$, because $w(u_2u_3) \le w(u_1u_2) + w(u_1u_3)$.

    \begin{figure}
        \centering
        \includegraphics{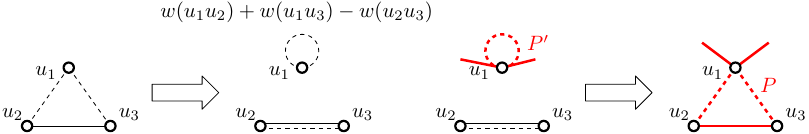}
        \caption{Reduction for Case~\ref{case:A12A2:red}.}
        \label{fig:A12A2}
    \end{figure}

    Since $A_1'$ and $A_2'$ are both $\calT'$-free $2$-matchings and $w'(A_2')=w(A_2) <(1-\eps)w(A_1)= (1-\eps)w'(A_1')$, one can apply the induction hypothesis to find an alternating trail $P'$ w.r.t.~$(A_1', A_2')$ such that $|P'| + 2\selfloop(P')\leq 7/\eps$, $w'(A_2'\Delta P')>w'(A_2')$ and $A_2'\Delta P'$ is a $\calT'$-free $2$-matching of $G'$. If $u_1u_1\in P'$, then let $P$ be the alternating trail resulting from inserting the trail $u_1u_2, u_2u_3, u_1u_3$ in place of $u_1u_1$ (see Figure~\ref{fig:A12A2} right). If $u_1u_1\notin P'$, then let $P:=P'$. Notice that $P$ is an alternating trail w.r.t.~$(A_1, A_2)$. Let us show that $P$ satisfies the conditions of the theorem. 

    Notice that $|P| + 2\selfloop(P)=|P'| + 2\selfloop(P')\leq 7/\eps$ holds regardless of whether $u_1u_1 \in P'$ or not.
    We now prove that $w(A_2\Delta P)>w(A_2)$. If $u_1u_1\notin P'$ then $A_2\Delta P = ((A_2'\Delta P')\setminus\{u_1u_1, u_2u_3\})\cup \{u_1u_2, u_1u_3\}$. Thus, 
    \begin{align*}
    w(A_2\Delta P)&=w'(A_2'\Delta P') - w'(u_1u_1) - w'(u_2u_3) + w(u_1u_2) + w(u_1u_3) \\
                  &= w'(A_2'\Delta P') > w'(A_2')=w(A_2).
    \end{align*}
    On the other hand, if $u_1u_1\in P'$, then $A_2\Delta P = A_2'\Delta P'$. Thus, $w(A_2\Delta P)=w'(A_2'\Delta P') > w'(A_2')=w(A_2)$. 
    In both cases, it holds that $w(A_2\Delta P)>w(A_2)$. It is left to show that $A_2\Delta P$ is a $\calT$-free $2$-matching.
    
    If $u_1u_1\notin P'$, then $A_2\Delta P = ((A_2'\Delta P')\setminus\{u_1u_1, u_2u_3\})\cup \{u_1u_2, u_1u_3\}$, while if $u_1u_1\in P'$, $A_2\Delta P = A_2'\Delta P'$. In both cases, it holds that $d_{A_2\Delta P}(u)=d_{A_2'\Delta P'}(u)\leq 2$ for all $u\in V$. Thus, $A_2\Delta P$ is a $2$-matching. 
    
    Since $A_2'\Delta P'$ is $\calT'$-free, $A_2\Delta P\subseteq(A_2'\Delta P')\cup \{u_1u_2, u_1u_3\}$, and neither $u_1u_2$ nor $u_1u_3$ is contained in a triangle in $\calT'$, one has that $A_2\Delta P$ is $\calT'$-free. Now, if $u_1u_1\notin P'$, then $A_2\Delta P = ((A_2'\Delta P')\setminus\{u_1u_1, u_2u_3\})\cup \{u_1u_2, u_1u_3\}$, and thus $(A_2\Delta P)\cap T = \{u_1u_2, u_1u_3\}$, so $A_2\Delta P$ is $\calT(\{u_1u_2, u_1u_3\})$-free by Observation~\ref{obs:2edges}. If $u_1u_1\notin P'$, then 
    $A_2\Delta P=A_2'\Delta P'$, and thus $A_2\Delta P$ contains neither $u_1u_2$ nor $u_1u_3$, implying that it is $\calT(\{u_1u_2, u_1u_3\})$-free.
    Therefore, $A_2\Delta P$ is $\calT$-free. The claim follows.

\bigskip
    \case{There exists a triangle $T\in \calT$ with $|A_1 \cap T| = 2$.}\label{case:final1} 

    Let $T=\{u_1u_2, u_2u_3, u_1u_3\}$.
    Since Cases~\ref{case:trivial},~\ref{case:A1A2}, and~\ref{case:2A1A2} do not hold, we can assume w.l.o.g.~that there exists a triangle $T'=\{u_2u_4, u_2u_3, u_3u_4\}\in\calT$ with $u_4 \in V \setminus \{u_1, u_2, u_3\}$, $\{u_1u_2, u_1u_3, u_2u_4, u_3u_4\}\subseteq A_1\setminus A_2$, and $u_2u_3\in A_2\setminus A_1$ (see Figures~\ref{fig:final1u1u4} and~\ref{fig:final1} left). 
    Let us first show the following. 

    \begin{claim}\label{clm:vertexSet}
        For every $T''\in \calT(T\cup T')$, the node set of $T''$ is contained in $\{u_1, u_2, u_3, u_4\}$.
    \end{claim}
    \begin{proof}
        Assume to get a contradiction that there exists a triangle $T''\in \calT(T\cup T')$ such that the node set of $T''$ contains a node $u\in V\setminus\{u_1, u_2, u_3, u_4\}$. Note that $T''$ is contained in $A_1\cup A_2$ because Case~\ref{case:trivial} does not hold. Let $T'' = \{xy, xu, yu\}$, where $xy\in\{u_1u_2, u_1u_3, u_2u_4, u_3u_4, u_2u_3\}$. If $xy = u_1u_2$, then $u_1u$ and $u_2u$ are in $A_2\setminus A_1$ due to the degree constraint of $A_1$. 
        Then, since $(A_2 \Delta T'') \cap T = \{u_1u_2, u_2u_3\}$, $A_2 \Delta T''$ is $\calT(\{u_1u_2\})$-free by Observation~\ref{obs:2edges}. This together with $A_2 \Delta T'' \subseteq A_2 \cup \{u_1u_2\}$ shows that $A_2 \Delta T''$ is $\calT$-free.
        Therefore, $T''$ satisfies the condition in Case~\ref{case:A12A2}, a contradiction. 
        The same argument applies when $xy\in\{u_1u_3, u_2u_4, u_3u_4\}$. If $xy=u_2u_3$, then $u_2u$ and $u_3u$ are in $A_2\setminus A_1$ due to the degree constraint of $A_1$. Then, $T''$ is contained in $A_2$, a contradiction to the assumption that $A_2$ is $\calT$-free.
    \end{proof}

    We consider the following subcases.
    
    \subcase{$u_1u_4\in A_2.$}\label{case:final1u1u4} Notice that $u_1u_4\notin A_1$, because of the degree constraints on $u_1$ and $u_4$. If $w(u_1u_2) + w(u_3u_4) > w(u_1u_4) + w(u_2u_3)$ then let $P:=\{u_1u_4, u_1u_2, u_2u_3, u_3u_4\}$. We claim that $P$ satisfies the conditions of the theorem. Indeed, since $w(u_1u_2) + w(u_3u_4) > w(u_1u_4) + w(u_2u_3)$, $w(A_2\Delta P) > w(A_2)$. One has $|P| + 2\selfloop(P) = 4 < 7/\eps$, and since $P$ is an even closed alternating trail, $A_2\Delta P$ is a $2$-matching. Moreover, by Claim~\ref{clm:vertexSet}, the only possible triangles in $A_2\Delta P$ are $T, T'$, $\{u_1u_2, u_2u_4, u_1u_4\}$, and $\{u_1u_3, u_3u_4, u_1u_4\}$, but none of them is in $A_2\Delta P$, so $A_2\Delta P$ is $\calT$-free. 
    
    From now on we assume that $w(u_1u_2) + w(u_3u_4) \leq w(u_1u_4) + w(u_2u_3)$. We set (see Figure~\ref{fig:final1u1u4}):

    \begin{itemize}
        \item $A_1':=(A_1\setminus\{u_1u_2, u_3u_4\})\cup\{u_1u_4, u_2u_3\}$, $A_2':=A_2$
        \item $w'(e):=w(e)$ for all $e\in E$,
        \item $\calT':=\calT\setminus\calT(T\cup T')$.
    \end{itemize}

    \begin{figure}
        \centering
        \includegraphics{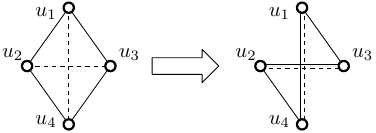}
        \caption{Reduction for Case~\ref{case:final1u1u4}.}
        \label{fig:final1u1u4}
    \end{figure}

    Notice that $A_1'$ and $A_2'$ are both $\calT'$-free $2$-matchings of $G'$. 
    Also, $w'(A_2')=w(A_2)$ and $w'(A_1')=w(A_1) - w(u_1u_2) - w(u_3u_4) + w(u_1u_4) + w(u_2u_3) \geq w(A_1)$ since $w(u_1u_2) + w(u_3u_4) \leq w(u_1u_4) + w(u_2u_3)$. Thus, by the assumption of the theorem, $w'(A_2')=w(A_2)<(1-\eps)w(A_1)\leq (1-\eps)w'(A_1')$. Let us apply the induction hypothesis to find an alternating trail $P'$ w.r.t.~$(A_1', A_2')$ such that $|P'| + 2\selfloop(P')\leq 7/\eps$, $w'(A_2'\Delta P')>w'(A_2')$, and $A_2'\Delta P'$ is a $\calT'$-free $2$-matching of $G'$. Let $P:=P'$, and notice that $P$ is an alternating trail w.r.t.~$(A_1, A_2)$. We claim that $P$ satisfies the conditions of the theorem.

    One has that $|P| + 2\selfloop(P)= |P'| + 2\selfloop(P') \leq 7/\eps$, and since $P=P'$ and $A_2=A_2'$, one also has that $w(A_2\Delta P)=w'(A_2'\Delta P')>w'(A_2')=w(A_2)$. Again because $P=P'$ and $A_2=A_2'$, and since $A_2'\Delta P'$ is a $2$-matching, $A_2\Delta P$ is a $2$-matching. It is left to show that $A_2\Delta P$ is $\calT$-free.

    Since $A_2'\Delta P'$ is $\calT'$-free, so is $A_2\Delta P$. By Claim~\ref{clm:vertexSet}, it is only left to show that $A_2\Delta P$ does not contain $T, T'$, $T'':=\{u_1u_4, u_3u_4, u_1u_3\}$, and $T''':=\{u_1u_4, u_2u_4, u_1u_2\}$. Since $u_1u_2\notin A_2\Delta P$, neither $T$ nor $T'''$ is in $A_2\Delta P$. Since $u_3u_4\notin A_2\Delta P$, neither $T'$ nor $T''$ is in $A_2\Delta P$. The claim follows.
    
    \subcase{$u_1u_4\notin A_2$.}\label{case:final1notu1u4} Notice that $u_1u_4\notin A_1$, because of the degree constraints on $u_1$ and $u_4$. We build an auxiliary graph $G' = (V', E', w')$ 
    by splitting each of $u_2$ and $u_3$ in a similar way to Case~\ref{case:A1A2}. 
    Let (see Figure~\ref{fig:final1} right):
    \begin{itemize}
        \item $V' = V\cup \{u_2', u_3'\}$, $E' = (E\setminus\{u_1u_2, u_3u_4, u_2u_3\})\cup\{u_1u_2', u_3'u_4, u_2'u_3', u_2u_2', u_3u_3'\}$, 
        \item $A_1':=(A_1\setminus\{u_1u_2, u_3u_4\})\cup\{u_1u_2', u_3'u_4, u_2u_2', u_3u_3'\}$, $A_2':=(A_2\setminus\{u_2u_3\})\cup\{u_2'u_3', u_2u_2', u_3u_3'\}$,
        \item $w'(e):=w(e)$ for all $e\in E\setminus\{u_1u_2, u_3u_4, u_2u_3\}$,
        \item $w'(u_1u_2')=w(u_1u_2)$, $w'(u_3'u_4)=w(u_3u_4)$, $w'(u_2'u_3')=w(u_2u_3)$, $w'(u_2u_2')=w'(u_3u_3')=0$,
        \item $\calT' = \calT\setminus\calT(T\cup T')$.
    \end{itemize}

    \begin{figure}
        \centering
        \includegraphics{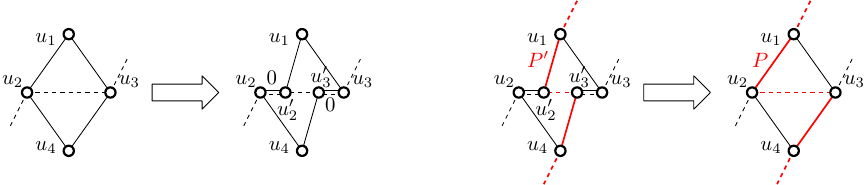}
        \caption{Reduction for Case~\ref{case:final1notu1u4}.}
        \label{fig:final1}
    \end{figure}

    Notice that $A_1'$ and $A_2'$ are both $2$-matchings of $G'$. Moreover, since $A_1$ and $A_2$ are $\calT$-free and the edges added to $A_1$ and $A_2$ are not in $G$, $A_1'$ and $A_2'$ are both $\calT'$-free. Notice as well that $w'(A_1') = w(A_1)$ and $w'(A_2') = w(A_2)$.
    
    Since $A_1'$ and $A_2'$ are both $\calT'$-free $2$-matchings and $w'(A_2')=w(A_2) < (1-\eps)w(A_1)= (1-\eps)w'(A_1')$, one can apply the induction hypothesis to find an alternating trail $P'$ w.r.t.~$(A_1', A_2')$ such that $|P'| + 2\selfloop(P')\leq 7/\eps$, $w'(A_2'\Delta P')>w'(A_2')$ and $A_2'\Delta P'$ is a $\calT'$-free $2$-matching of $G'$. Let $P$ be the trail resulting from replacing $u_1u_2'$ by $u_1u_2$, $u_2'u_3'$ by $u_2u_3$, and $u_3'u_4$ by $u_3u_4$, if any of them is present in $P'$ (see Figure~\ref{fig:final1} right). Notice that $P$ is an alternating trail w.r.t.~$(A_1, A_2)$. We claim that $P$ satisfies the conditions of the theorem. 

    Notice that $|P| + 2\selfloop(P)=|P'| + 2\selfloop(P')\leq 7/\eps$. We now prove that $w(A_2\Delta P)>w(A_2)$. 
    Observe that $A_2$ (resp.,~$A_2\Delta P$) in $G$ is obtained from $A'_2$ (resp.,~$A'_2 \Delta P'$) in $G'$ by contracting both $u_2 u'_2$ and $u_3 u'_3$.
    Since $w(u_1u_2) = w'(u_1'u_2)$, $w(u_2u_3) = w'(u_2'u_3')$, $w(u_3u_4) = w'(u_3'u_4)$, and $w'(u_2u_2') = w'(u_3u_3') = 0$, this observation shows that $w(A_2\Delta P) = w'(A_2'\Delta P')>w'(A_2')=w(A_2)$.
    
    Let us now show that $A_2\Delta P$ is a $2$-matching of $G$. 
    Since $A_2\Delta P$ is obtained from $A'_2 \Delta P'$ by contracting both $u_2 u'_2$ and $u_3 u'_3$, for every node $u\in V\setminus\{u_2, u_3\}$, $d_{A_2\Delta P}(u) = d_{A_2'\Delta P'}(u)\leq 2$. Since the edge $u_2u_2'$ counts toward both $d_{A_2'\Delta P'}(u_2)$ and $d_{A_2'\Delta P'}(u_2')$, but is not present in $A_2\Delta P$, one has that $d_{A_2\Delta P}(u_2) = d_{A_2'\Delta P'}(u_2)-1+d_{A_2'\Delta P'}(u_2')-1\leq 2$. A similar argument works for $u_3$, so that $d_{A_2\Delta P}(u)\leq 2$ for every $u\in V$.

    Finally, it is left to show that $A_2\Delta P$ is $\calT$-free. Since $A_2'\Delta P'$ is $\calT'$-free, so is $A_2\Delta P$. We now show that $A_2\Delta P$ is $\calT(T\cup T')$-free, which implies that it is $\calT$-free. By Claim~\ref{clm:vertexSet} and by this case assumption $u_1u_4\notin A_1\cup A_2$, the only triangles in $\calT$ that are not in $\calT'$ are $T$ and $T'$. Notice that if $u_1u_2'\in P'$, then $u_2'u_3'\in P'$, because of the degree constraints on $u_2'$. This implies that if $u_1u_2\in P$ then $u_2u_3\in P$. In turn this implies that, if $u_1u_2\in A_2\Delta P$, then $u_2u_3\notin A_2\Delta P$. So $T$ is not in $A_2\Delta P$. Similarly, if $u_3u_4\in A_2\Delta P$, then $u_2u_3\notin A_2\Delta P$, so $T'$ is not in $A_2\Delta P$. The claim follows.

    \bigskip
    \case{There exists a triangle $T\in \calT$ with $|A_2 \cap T| = 2$.}\label{case:final2}  

    Let $T=\{u_1u_2, u_2u_3, u_1u_3\}$.   
    Since Cases~\ref{case:trivial},~\ref{case:A1A2}, and~\ref{case:A12A2} do not hold, we can assume w.l.o.g. that there exists a triangle $T'=\{u_2u_4, u_2u_3, u_3u_4\}$ with $u_4 \in V \setminus \{u_1, u_2, u_3\}$, $\{u_1u_2, u_1u_3, u_2u_4, u_3u_4\}\subseteq A_2\setminus A_1$, and $u_2u_3\in A_1\setminus A_2$ (see Figures~\ref{fig:final2u1u4} and~\ref{fig:final2} left). We will use an analogous claim as in the previous case.

    \begin{claim}\label{clm:vertexSet2}
        For every $T''\in \calT(T\cup T')$, the node set of $T''$ is contained in $\{u_1, u_2, u_3, u_4\}$.
    \end{claim}
    \begin{proof}
    The proof is symmetric to the proof of Claim~\ref{clm:vertexSet}.
    \end{proof}

    We consider the following subcases.

    \subcase{$u_1u_4\in A_1.$}\label{case:final2u1u4} Notice that $u_1u_4\notin A_2$, because of the degree constraints on $u_1$ and $u_4$. If $w(u_1u_4) + w(u_2u_3) > w(u_1u_2) + w(u_3u_4)$ then let $P:=\{u_1u_4, u_1u_2, u_2u_3, u_3u_4\}$. We claim that $P$ satisfies the conditions of the theorem. Indeed, since $w(u_1u_4) + w(u_2u_3) > w(u_1u_2) + w(u_3u_4)$, $w(A_2\Delta P) > w(A_2)$. One has $|P| + 2\selfloop(P)= 4 < 7/\eps$, and since $P$ is an even closed alternating trail, $A_2\Delta P$ is a $2$-matching. Moreover, by Claim~\ref{clm:vertexSet2}, the only possible triangles in $A_2\Delta P$ are $T, T'$, $\{u_1u_2, u_2u_4, u_1u_4\}$, and $\{u_1u_3, u_3u_4, u_1u_4\}$, but none of them is in $A_2\Delta P$, so $A_2\Delta P$ is $\calT$-free. 
    
    From now on we assume that $w(u_1u_4) + w(u_2u_3) \leq w(u_1u_2) + w(u_3u_4)$. We set (see Figure~\ref{fig:final2u1u4}):

    \begin{itemize}
        \item $A_1':=A_1$, $A_2':=(A_2\setminus\{u_1u_2, u_3u_4\})\cup\{u_1u_4, u_2u_3\}$
        \item $w'(e):=w(e)$ for all $e\in E\setminus\{u_2u_3, u_1u_4\}$,
        \item $w'(u_2u_3) = w(u_1u_2)$, $w'(u_1u_4)=w(u_3u_4)$,
        \item $\calT':=\calT\setminus\calT(T\cup T')$.
    \end{itemize}

    \begin{figure}
        \centering
        \includegraphics{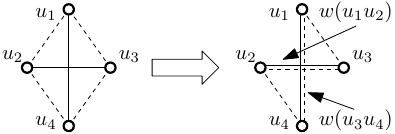}
        \caption{Reduction for Case~\ref{case:final2u1u4}.}
        \label{fig:final2u1u4}
    \end{figure}

    Notice that $A_1'$ and $A_2'$ are both $\calT'$-free $2$-matchings of $G'$.  
    Also, $w'(A_2')=w(A_2)$ and $w'(A_1')=w(A_1) - w(u_1u_4) - w(u_2u_3) + w(u_1u_2) + w(u_3u_4)\geq w(A_1)$ since $w(u_1u_4) + w(u_2u_3) \leq w(u_1u_2) + w(u_3u_4)$. Thus, by the assumption of the theorem, $w'(A_2')=w(A_2)<(1-\eps)w(A_1)\leq (1-\eps) w'(A_1')$. Thus one can apply the induction hypothesis to find an alternating trail $P'$ w.r.t.~$(A_1', A_2')$ such that $|P'| + 2\selfloop(P')\leq 7/\eps$, $w'(A_2'\Delta P')>w'(A_2')$, and $A_2'\Delta P'$ is a $\calT'$-free $2$-matching of $G'$. Let $P:=P'$, and notice that $P$ is an alternating trail w.r.t.~$(A_1, A_2)$. We claim that $P$ satisfies the conditions of the theorem.

    One has that $|P| + 2\selfloop(P) = |P'| + 2\selfloop(P')\leq 7/\eps$. Notice that $A_2\Delta P = ((A_2'\Delta P') \setminus\{u_1u_4, u_2u_3\})\cup \{u_1u_2, u_3u_4\}$. Since $w'(u_1u_4)=w(u_1u_2)$ and $w'(u_2u_3)=w(u_3u_4)$, one has that $w(A_2\Delta P)=w'(A_2'\Delta P')>w'(A_2')=w(A_2)$. Again, since $A_2\Delta P = ((A_2'\Delta P') \setminus\{u_1u_4, u_2u_3\})\cup \{u_1u_2, u_3u_4\}$, $d_{A_2\Delta P}(u)=d_{A_2'\Delta P'}(u)\leq 2$ for every $u\in V$, and thus $A_2\Delta P$ is a $2$-matching. It is left to show that $A_2\Delta P$ is $\calT$-free.

    Since $A_2'\Delta P'$ is $\calT'$-free, so is $A_2\Delta P$. By Claim~\ref{clm:vertexSet2} it is only left to show that $A_2\Delta P$ does not contain $T, T'$, $T'':=\{u_1u_4, u_3u_4, u_1u_3\}$, and $T''':=\{u_1u_4, u_2u_4, u_1u_2\}$. Since $u_2u_3\notin A_2\Delta P$, neither $T$ nor $T'$ is in $A_2\Delta P$. Since $u_1u_4\notin A_2\Delta P$, neither $T$ nor $T'$ is in $A_2\Delta P$. The claim follows.

    \subcase{$u_1u_4\notin A_1$.}\label{case:final2notu1u4} Notice that $u_1u_4\notin A_2$, because of the degree constraints on $u_1$ and $u_4$. We build an auxiliary graph $G' = (V', E', w')$ 
    by splitting each of $u_2$ and $u_3$ in the same way as in Case~\ref{case:final1notu1u4}.
    Let (see Figure~\ref{fig:final2} left):
    \begin{itemize}
        \item $V' = V\cup \{u_2', u_3'\}$, $E' = (E\setminus\{u_1u_2, u_3u_4, u_2u_3\})\cup\{u_1u_2', u_3'u_4, u_2'u_3', u_2u_2', u_3u_3'\}$, 
        \item $A_1':=(A_1\setminus\{u_2u_3\})\cup\{u_2'u_3', u_2u_2', u_3u_3'\}$, $A_2':=(A_2\setminus\{u_1u_2, u_3u_4\})\cup\{u_1u_2', u_3'u_4, u_2u_2', u_3u_3'\}$,
        \item $w'(e):=w(e)$ for all $e\in E\setminus\{u_1u_2, u_3u_4, u_2u_3\}$,
        \item $w'(u_1u_2')=w(u_1u_2)$, $w'(u_3'u_4)=w(u_3u_4)$, $w'(u_2'u_3')=w(u_2u_3)$, $w'(u_2u_2')=w'(u_3u_3')=0$,
        \item $\calT' = \calT\setminus\calT(T\cup T')$.
    \end{itemize}

    \begin{figure}
        \centering
        \includegraphics{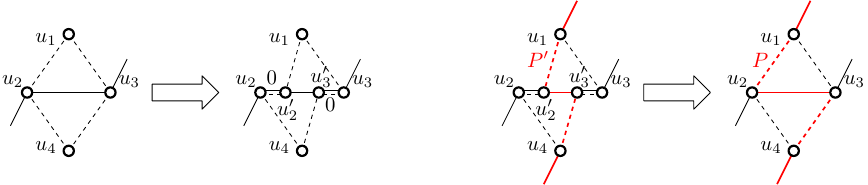}
        \caption{Reduction for Case~\ref{case:final2notu1u4}.}
        \label{fig:final2}
    \end{figure}

    Notice that $A_1'$ and $A_2'$ are both $2$-matchings of $G'$. Moreover, since $A_1$ and $A_2$ are $\calT$-free, and the edges added to $A_1$ and $A_2$ are not in $G$, $A_1'$ and $A_2'$ are both $\calT'$-free. 
    Notice as well that $w'(A_1') = w(A_1)$ and $w'(A_2') = w(A_2)$.
    
    Since $A_1'$ and $A_2'$ are both $\calT'$-free $2$-matchings and $w'(A_2')=w(A_2) < (1-\eps)w(A_1)= (1-\eps)w'(A_1')$, one can apply the induction hypothesis to find an alternating trail $P'$ w.r.t.~$(A_1', A_2')$ such that $|P'| + 2\selfloop(P')\leq 7/\eps$, $w'(A_2'\Delta P')>w'(A_2')$, and $A_2'\Delta P'$ is a $\calT'$-free $2$-matching of $G'$. Let $P$ be the trail resulting from replacing $u_1u_2'$ by $u_1u_2$, $u_2'u_3'$ by $u_2u_3$, and $u_3'u_4$ by $u_3u_4$, if any of them is present in $P'$ (see Figure~\ref{fig:final2} right). Notice that $P$ is an alternating trail w.r.t.~$(A_1, A_2)$. We claim that $P$ satisfies the conditions of the theorem. 

    Notice that $|P| + 2\selfloop(P)=|P'| + 2\selfloop(P')\leq 7/\eps$. We now prove that $w(A_2\Delta P)>w(A_2)$. 
    Observe that $A_2$ (resp.,~$A_2\Delta P$) in $G$ is obtained from $A'_2$ (resp.,~$A'_2 \Delta P'$) in $G'$ by contracting both $u_2 u'_2$ and $u_3 u'_3$.
    Since $w(u_1u_2) = w'(u_1'u_2)$, $w(u_2u_3) = w'(u_2'u_3')$, $w(u_3u_4) = w'(u_3'u_4)$, and $w'(u_2u_2') = w'(u_3u_3') = 0$, this observation shows that $w(A_2\Delta P) = w'(A_2'\Delta P')>w'(A_2')=w(A_2)$.
    
    Let us now show that $A_2\Delta P$ is a $2$-matching of $G$. 
    Since $A_2\Delta P$ is obtained from $A'_2 \Delta P'$ by contracting both $u_2 u'_2$ and $u_3 u'_3$, for every node $u\in V\setminus\{u_2, u_3\}$, $d_{A_2\Delta P}(u) = d_{A_2'\Delta P'}(u)\leq 2$. Since the edge $u_2u_2'$ counts toward both $d_{A_2'\Delta P'}(u_2)$ and $d_{A_2'\Delta P'}(u_2')$, but is not present in $A_2\Delta P$, one has that $d_{A_2\Delta P}(u_2) = d_{A_2'\Delta P'}(u_2)-1+d_{A_2'\Delta P'}(u_2')-1\leq 2$. A similar argument works for $u_3$, so that $d_{A_2\Delta P}(u)\leq 2$ for every $u\in V$.
        
    Finally, it is left to show that $A_2\Delta P$ is $\calT$-free. Since $A_2'\Delta P'$ is $\calT'$-free, so is $A_2\Delta P$. We now show that $A_2\Delta P$ is $\calT(T\cup T')$-free, which implies that it is $\calT$-free. By Claim~\ref{clm:vertexSet2} and by this case assumption $u_1u_4\notin A_1\cup A_2$, the only triangles in $\calT$ that are not in $\calT'$ are $T$ and $T'$. Notice that if $u_2'u_3'\in P'$ then both $u_1u_2'\in P'$ and $u_3'u_4\in P'$ hold, because of degree constraints on $u_2'$ and $u_3'$. This implies that if $u_2u_3\in P$ then $u_1u_2, u_3u_4\in P$, which in turn implies that if $u_2u_3\in A_2\Delta P$, then $u_1u_2, u_3u_4\notin A_2\Delta P$, so neither $T$ nor $T'$ is in $A_2\Delta P$. The claim follows. \qed
       
    \end{caseanalysis}

\section*{Acknowledgments}
This work was supported by the joint project of Kyoto University and Toyota Motor Corporation,
titled ``Advanced Mathematical Science for Mobility Society'', by JSPS KAKENHI Grant Numbers JP22H05001 and JP24K02901, and by JST SPRING, Grant Number JPMJSP2110. The first 2 authors are partially supported by the SNF Grant $200021\_200731 / 1$.

\bibliographystyle{abbrv}
\bibliography{ref}

\end{document}